\documentclass[10pt,twocolumn]{article}
\usepackage{geometry}
\usepackage{amsthm,amsmath,amssymb}
\usepackage{graphicx}
\usepackage{color}
\usepackage[ruled,vlined]{algorithm2e}

\def\Pr{\mbox{{\bf P}}}
\def\Prob{\Pr}
\def\polylog{\operatorname{polylog}}

\newcommand{\ignore}[1]{}

\newtheorem{theorem}{Theorem}
\newtheorem{lemma}[theorem]{Lemma}

\newtheorem{proposition}[theorem]{Proposition}
\newtheorem{remark}[theorem]{Remark}

\makeatletter
\newcommand{\mytag}[2]{%
	\text{#1}%
	\@bsphack
	\begingroup
	\@onelevel@sanitize\@currentlabelname
	\edef\@currentlabelname{%
		\expandafter\strip@period\@currentlabelname\relax.\relax\@@@%
	}%
	\protected@write\@auxout{}{%
		\string\newlabel{#2}{%
			{#1}%
			{\thepage}%
			{\@currentlabelname}%
			{\@currentHref}{}%
		}%
	}%
	\endgroup
	\@esphack
}
\makeatother

\begin{document}

\title{Best-of-Three Voting on Dense Graphs \thanks{Nicolas Rivera is supported by Thomas Sauerwald's ERC Starting Grant 679660 (DYNAMIC MARCH)
}
}

\author{
Nan Kang\thanks{Department of Informatics, King's College London, UK.
{\tt nan.kang@kcl.ac.uk}}
\and Nicol\'as Rivera\thanks{Computer Laboratory, University of Cambridge, UK.
{\tt nicolas.rivera@cl.cam.ac.uk}}
}
\date{\today}

\maketitle

\begin{abstract}
Given a graph $G$ of $n$ vertices, where each vertex is initially attached an opinion of either red or blue. We investigate a random process known as the Best-of-three voting. In this process, at each time step, every vertex chooses three neighbours at random and adopts the majority colour. We study this process for a class of graphs with minimum degree $d = n^{\alpha}$\,, where $\alpha = \Omega\left( (\log \log n)^{-1} \right)$. We prove that if initially each vertex is red with probability greater than $1/2+\delta$, and blue otherwise, where $\delta \geq (\log d)^{-C}$ for some $C>0$, then with high probability this dynamic reaches a final state where all vertices are red  within $O\left( \log \log n\right) + O\left( \log \left( \delta^{-1} \right) \right)$ steps\,.

\end{abstract}

\section{Introduction}

Algorithms and protocols that solve consensus problems play an important role in distributed computing, analysis of social networks, etc. Usually, in these processes, vertices of a graph revise their opinions in a systematic and distributed way based on opinions of other neighbours, typically by sampling some of their neighbours. The aim of these protocols is to eventually reach a state where all vertices share the same opinion, and ideally this final state reflects the characteristics of the initial mix of opinions, e.g. the initial majority.

Among the protocols that solve consensus problems, one well-known protocol is the Best-of-$k$ model, in particular the cases with $k=1,2$ and $3$. In this protocol, we consider a graph $G = (V,E)$, in which each vertex has an initial opinion (colour), and at each time step, every vertex adopts the opinion of the majority of a sample of $k$ neighbours (uniformly with replacement). If there is no clear majority several rules can be applied, but usually i) the vertex keeps its opinion or ii) the vertex picks a random one from the popular opinions among the $k$ neighbours.

The Best-of-$1$ model is the well-known voter model. This protocol solves the consensus problem in connected non-bipartite graphs. It is widely known that the probability of the system reaching consensus on a particular colour is proportional to the sum of the degrees of the vertices whose initial opinion is such colour. Specifically, a particular colour `wins' with probability equals to the initial proportion of that colour in the configuration. Although the voter model can be used to solve consensus problems, it is not the desired protocol for applications where consensus to majority is required.  

Best-of-$k$ with $k=2,3$ partially overcomes the aforementioned problem, and it converges to majority under appropriate circumstances. Moreover, it converges considerably faster compared to the voter model. This model has been extensively studied when the underlying topology is a complete graph. In \cite{Becchetti2014Simple}, the authors investigated the Best-of-$3$ dynamics breaking ties at random. They considered $q$ initial different opinions and proved that if the initial imbalance in the number of opinions between the first and second majorities is $\Omega(\min\{\sqrt{2q}, (n/\log n)^{1/6}\}\sqrt{n/\log n})$, then consensus is reached with high probability in $
O(\min\{q,(n/ \log n)^{1/3}\}\log n)$ steps on the initial majority. Similar results can be proved for Best-of-$2$ dynamics \cite{Ghaffari2018nearly}.

When it comes to non-complete graphs, the Best-of-$k$ process seems very difficult to study. In this regard, \cite{Cooper2014Power} studied the Best-of-two process on a $d$-regular graph where each vertex has one out of two opinions, say, Red or Blue. The authors showed that if the imbalance between the number of red and blue opinions is greater than $Kn\sqrt{1/d+d/n}$ initially, where $K$ is a large constant, then w.h.p the process reaches consensus towards majority in $O(\log n)$ time-steps. In \cite{Cooper2015Fast}, the result was extended and refined to general graphs with large expansion. Denote by $R_0$ and $B_0$ the initial sets of vertices with red and blue opinions respectively. Then assume that $d(R_0)- d(B_0) \geq  4\lambda_2 d(V)$, where $d(X)$ denotes the sum of the degrees of the vertices in $X$, and $\lambda_2$ is the second largest absolute eigenvalue of the transition matrix associated with the graph, then w.h.p consensus is reached in $O(\log n)$ rounds and opinion red wins. In regular graphs, the aforementioned condition implies an $\Omega(n)$ gap between the sizes of the set $R_0$ and $B_0$. In \cite{Cooper2017Fast}, the result is extended to a larger number of initial opinions but with stronger assumptions.

Best-of-$k$ with odd $k \geq 5$ was studied in \cite{Abdullah2015Global} for the two-party model on random graphs with a given degree sequence. Under their setting, initially, each vertex is blue independently with probability $1/2 - \delta >0$, and red otherwise. Then, it is demonstrated that if $\delta$ is large enough, and $k \geq \hat{d}_{\min}$, then consensus is reached in $O(\log_k \log_k n)$ time steps and opinion red wins. Here $\hat{d}_{\min}$ is the effective minimum degree, which is the smallest integer that appears $\Theta(n)$ times in the degree sequence.

\subsection{Main Results}
In the current work, we study the particular case $k=3$ in the two-party setting, where initially each vertex is blue independently with probability $1/2-\delta >0$, otherwise red. By applying two models to analyse the process of a vertex updating its opinions, we find conditions for convergence to majority in $O(\log \log n)$ time-steps with high probability. Our main result is the following.

\begin{theorem}
Given a graph $G$ on $n$ vertices with minimum degree $d = n^{\alpha}$\,, where $\alpha = \Omega((\log \log n)^{-1})$\,, we suppose that initially each vertex is blue independently with probability $1/2 - \delta$\,, otherwise red, with $\delta \geq (\log d)^{-C}$ for some $C>0$. Then, w.h.p, the Best-of-Three protocol reaches consensus in $O(\log \log n) + O\left(\log(\delta^{-1})\right)$ time-steps and the final opinion is red.
\end{theorem}

Compared to previous work, \cite{Abdullah2015Global} is the closest to ours, as they also look at forward conditions to ensure double logarithmic consensus time towards majority and they work on non-complete graphs. In order to reach double logarithmic speed the graph requires a tractable local structure around each vertex, and we need to be able to keep track of the configuration of opinions around each vertex at each time step. In this regard, the techniques used in 
\cite{Cooper2014Power} and \cite{Cooper2015Fast} are not necessarily useful to tackle the problem in our case, even though they work on a large class of graphs. This is because in their work, the authors track the number of red (and blue) opinions instead of the actual configuration of the opinions of the vertices. Although tracking the number of red opinions is easier, the obtained result is not precise enough, and indeed, the technique gives $O(\log n)$ steps towards consensus, which is not fast enough as desired. Moreover, by tracking only the number of red vertices, we lose the extra information of how the opinions of vertices are distributed given by the fact that vertices start with randomised opinions. Additionally, the proof technique used in \cite{Cooper2015Fast} works under adversarial setting where the adversarial can reorganise the opinions among the vertices and keep the total number of each opinion fixed, thus the initial location of the opinion does not matter.

With respect to \cite{Abdullah2015Global}, our result is weaker in some respects and much stronger in others. First of all, both proofs are based on a sort of `time-reversal duality' while instead of tracking the opinion of a vertex $v$ from time $0$ to fixed value $T$\,, we obtain the opinion at time $T$ by looking at the opinions at time $T-1$, and to determine those we look at the opinions at time $T-2$ etc. The process of keeping track of the opinion of a vertex is more complex, since it depends on several random variables which are dependent and thus difficult to analyse. To avoid dealing with such problem directly, \cite{Abdullah2015Global} decided to work in the setting of $k \geq 5$\,, which allows them to assume that certain vertices have the `bad' opinion (i.e. minority) even if they actually have the `good' opinion (majority). This helps them to reduce the dependency caused by the opinion updating process so as to transform the real process into a simpler and easier-to-analyse process. As $k \geq 5$, assuming that one opinion is `bad' does not particularly damage the speed of convergence to consensus, as we can hope that the other $k-1$ opinions have the `good' majority. However, since for some vertex one `bad' opinion is assumed, they rely on the other $k-1$ opinions getting the good majority quite often in the process. In other words, they need to ensure a large initial gap between the initial numbers of the two opinions, thus their result holds only when the initial probability of being blue is way less than $1/2$ (i.e. $1/2-q$ for large enough $q \in (0,1/2)$). Due to this reason, their result cannot be extended to $k=3$\,, as assuming a `bad' opinion will affect the majority significantly. Indeed, if one of the other two opinions is the `bad' opinion, then the vertex will adopt it. Our proof partially overcomes this problem. We work with $k=3$ and allow the initial probability of being blue to be $1/2-\delta$ , where $\delta$ is arbitrarily close to $0$ and  we can even choose it tending to $0$ as the graph grows. Finally, our analysis works on the family of graphs with minimum degree $n^{\Omega(1/\log \log n)}$,  whereas in \cite{Abdullah2015Global} the authors consider random graphs of a given degree sequence with average degree $o(\log n)$ among other constrains. Note that both classes of graphs are disjoint.

\section{Model and Proof Strategy}

Let us recall our model and introduce some notations. Let $G = (V,E)$ be a graph where each vertex is blue ($\textbf{B}$) independently with probability $1/2-\delta$, otherwise red ($\textbf{R}$) \,.

Define the opinion set of each vertex at time $t$ to be $\xi_t = \left( \xi_t(1), \cdots, \xi_t(n) \right) \,.$ The evolution of the opinions $(\xi_t)_{t \geq 0}$ is as follows: For $v \in V$, define $\xi_0(v)$ as the initial opinion of $v$. For each $t \geq 0$, every vertex $v$ independently samples three random neighbours $w_v^1, w_v^2, w_v^3$ (with replacement), and sets $\xi_{t+1}(v) = \text{majority}\{\xi_t(w_v^1), \xi_t(w_v^2), \xi_t(w_v^3)\}$\,. Note that the value of $(\xi_{t+1}(v))_{v \in V}$ is determined only by $(\xi_{t}(v))_{v \in V}$ plus some independent randomness, i.e. $(\xi_t)_{t \geq 0}$ is a Markov chain. 

Our proof strategy consists of verifying that $\Prob(\xi_T(v) = \textbf{B}) = o(1/n)$ holds for $T = O(\log \log n)$\,, and thus $\Prob(G\text{ is red at step $T$}) =1- \Prob(\bigcup_{v \in V}\{\xi_T(v)=\textbf{B}\}) = 1-o(1)$. Therefore, the emphasis of our work is essentially in computing $\Prob(\xi_T(v) = \textbf{B})$. From the definition of the process we know that $\xi_T(v)$ is determined by the opinions of three random neighbours of vertex $v$\,, say $w_v^1, w_v^2, w_v^3$\,, at time $T-1$, i.e. $\xi_{T-1}(w_v^i)$. Similarly, $\xi_{T-1}(w_v^i)$ is determined by the opinions at time $T-2$ of three random neighbours of $w_v^i$. We can continue recursively until the point where we query for the opinions of vertices at time $0$, whose joint distribution is known. The above recursive (random) structure can be represented as a directed acyclic graph (DAG). 

A DAG $H$ is a directed graph with no directed cycles. The in-degree of a vertex $v$ is the number of edges incoming to $v$, while the out-degree is the number of edges outgoing from $v$. A root in $V(H)$ is a vertex with in-degree 0. In this work we assume that there is only 1 root. A leaf in $V(H)$ is a vertex with out-degree $0$. Given $v \in V(H)$, we define $H[v]$ as the subgraph induced by all the vertices $w$ that can be reached from $v$, i.e. there exists a directed path from $v$ to $w$. As in this work we will consider some random DAGs, we shall denote them by $\mathcal H$ while $H$ is used to denote a fixed, deterministic DAG.

Let us construct the random voting-DAG associated with $\xi_T(v_0)$. In our work, we call it a voting-DAG to specify the DAG that has out-degree at most three. Define the set $Q_T = \{v_0\}$, and for $t \in \{0,\ldots, T-1\}$ define $Q_t \subseteq V$ as the (random) subset of all vertices queried to determine the opinions of the vertices in $Q_{t+1}$ at time $t+1$\,, e.g. $Q_{T-1}$ is the set of three random neighbours of $v_0$ required to determine $\xi_T(v_0)$, etc. We define the random voting-DAG $\mathcal H = \mathcal H_{v_0}$ by setting $V(\mathcal H) = \{(v,t) \in V(G) \times \{0,\ldots, T\}: v \in Q_t\}$, and we say that $((v,t+1), (w,t)) \in E(\mathcal H)$ if and only if $(v, t+1) \in Q_{t+1}$ and one of the three vertices sampled by $v$ to compute $\xi_{t+1}(v)$ is $w$. Given the random voting-DAG $\mathcal H$ we divide its sets of vertices into levels, where level $t\in \{0,\ldots, T\}$ contains all the vertices $(v,t) \in V(\mathcal H)$. Note that each vertex at level $t+1$ connects to exactly three vertices at level $t$ and that directed paths go from higher to lower levels. 

Given a realisation $H$ of  $\mathcal H$ with root $(v_0,T)$ we can simulate $\xi_T(v)$ as following. First, settle the opinion of vertices $(v,0) \in Q_0$ to be independently \textbf{B} with probability $1/2 -\delta$, otherwise \textbf{R}\,. Then recursively compute the opinions of vertices at level $t+1$ as the majority of the three neighbours at level $t$\,, for $t \in \{0,\cdots,T-1\}$\,. Denote by $X_{H}(v,t)$ the colour of vertex $(v,t)$ in $H$. By summing up over all possible realisations of $\mathcal H$, it is clear that the colour of $(v_0,T)$ has the same distribution as $\xi_T(v_0)$, i.e.
$$\Prob\left(\xi_T(v_0)=\textbf{B}\right) = \Prob\left(X_{\mathcal H}(v_0,T) = \textbf{B}\right)\,.$$
Note that $X_{\mathcal H}(v_0,T)$ involves two independent sources of randomness. One source generates the voting-DAG $\mathcal H$, and the other settles the colours of the leaves of $\mathcal H$ (vertices at level 0) independently. Note that $(v_0,T) \in V(H)$ for any realisation $H$ of $\mathcal H$, so the random variable $X_{\mathcal H}(v_0,T)$ is well-defined. Finally, the process of defining $X_H$ as above to colour the realisation $H$ of $\mathcal H$ is referred as the colouring process. 

Given $\mathcal H = H$ we have that the opinions of the vertices at level $0$ are i.i.d,  as they do not depend on the structure imposed by $H$\,, which is unfortunately not true for levels $t>0$\,. Recall that $H[(v,t)]$ is the subgraph induced by all the reachable vertices from vertex $v$ at level $t$\,. Indeed, it is clear that the colour $X_H(v,t)$ depends only on the colouring of the leaves of $H[(v,t)]$\,. Therefore the variables $X_H(v,t)$ and $X_H(w,t)$ with $(v,t),(w,t) \in V(H)$ are independent if and only if $V(H[(v,t)])  \cap V(H[(w,t)])  = \emptyset$. As the structure distribution of $\mathcal H$ (and so its sample $H$) strongly depends on the underlying structure of $G$, it is very unlikely to have the above independence condition for all vertices in $H$. However, let us assume for a moment that the independence condition above is satisfied for all pairs of vertices sharing the same level. In such a case, $H$ is a directed ternary tree with root $(v_0,T)$\,. Let $B_{t-1}$ be the number of blue vertices at level $t-1$ among three random samples $Q_{t-1}(v)$ of a vertex $v$ at level $t-1$\,, and $b_t$ be the probability that any vertex at level $0<t\leq T$ is blue. Observe that $B_t$ follows a Binomial distribution $Bin(3, b_t)$\,, thus the probability $b_t$ follows the recursions: $b_0 = 1/2 - \delta$\,, and  
\begin{align}
b_{t} &= \Pr(B_{t-1}\ge 2) \nonumber \\
&=b_{t-1}^3+ 3b_{t-1}^2(1-b_{t-1}) = 3b_{t-1}^2-2b_{t-1}^3 \,.\label{eqn:basicRecursionK_n}
\end{align}
Therefore, a simple computation shows that by choosing $T = O(\log \log n+\log \delta^{-1})$ we get $b_T = o(n^{-1}).$

As the probability that $\mathcal H$ is realised as a ternary tree is low, the above recursion does not necessarily reflect the true process. In order to deal with the inner dependency between levels, we divide the graph $\mathcal H$ into two subgraphs, one from level $T$ to $T'$ and another from level $T'$ to $0$, where $T'$ is going to be fixed later. For the subgraph from level $T'$ to $0$ we couple the colouring $X_H$ with another colouring $X'_H$ such that if colour \textbf{B} represents $1$ and colour \textbf{R} represents 0, then $X_H(v,t) \leq X'_H(v,t)$  for all $(v,t) \in V(H)$. The process where $X'_H$ arises is called the Sprinkling process. By introducing an error term to deal with the dependency in this process, we have an easier way to study $X'_H$ as opinions among vertices at the same level are independent given $H$\,. Moreover, if we do not reveal $H$ in advance (i.e. randomize over $H$), the distribution of the colours of vertices at level $t$ for $t \in \{0, \cdots , T'\}$ is i.i.d. and follows a recursion similar to~\eqref{eqn:basicRecursionK_n}. Unfortunately, this recursion cannot be applied any further for $t>T'$, since after $T'$ steps it reaches a fixed point where the error term becomes significant in the recursion. Nevertheless, the above strategy is good enough to prove that with probability $1-o(n^{-1})$ the number of blue vertices at level $T'$ is sufficiently small. As for the subgraph of $H$ from $T'$ to $0$\,, the only way that the root of $H$ gets colour \textbf{B} is that the structure of $H$ from level $T$ to $T'$ is particularly bad for \textbf{R}\,. We can prove that the event of $H$ having such a structure occurs with probability $o(n^{-1})$\,.

\begin{remark}
The random voting-DAG $H(v_0)$ can be viewed as the trajectory of a Coalescing and Branching random walk or, for short, COBRA walk (see \cite{Berenbrink2018tight},\cite{Cooper2017Improved},  \cite{Mitzenmacher2018Better} for recent research). A COBRA walk is a discrete process on a graph $G$ where vertices are occupied by particles. At each time-step, each particle makes $k-1$ copies of itself and they locate at the same vertex, then all the particles in the graph independently move to a random neighbour. After that, if a vertex is occupied by more than one particle they coalesce into one. The process keeps repeating forever. In our setting, $H$ represents the trajectory of $T$ steps of a COBRA walk with $k=3$ starting with one particle, located at $v_0$. Level $T-t$ of $H$ represents the set of occupied vertices at time $t$ of the COBRA walk, and the edges between level $T-t$ and $T-t-1$ represent the movements of the particles between times $t$ and $t+1$. The COBRA walk with parameter $k =1$ is the classic Coalescing random walk process which is the dual process of the voter model (or best-of-1 according to our notation).
\end{remark}

The proof is presented in two parts. In Section\ref{sec:lowerLevels} we work with the lower levels of the voting-DAG (i.e. closer to the leaves) while in Section~\ref{sec:upperLevels} we study the colouring structure close to the root.

\section{Lower Levels}\label{sec:lowerLevels}

Let $G = (V,E)$ be a graph with minimum degree $d=n^{\beta/\log \log n}$ with $\beta>0$. For the simplicity of the results in this section we associate the opinion $\textbf{B}$ to the value $1$ and $\textbf{R}$ to the value $0$\,.

%

Let $T'\leq T$, and consider the following protocol, which is called the Sprinkling process. Suppose we only know the structure of the voting-DAG from level $0$ up to level $T'$\,, then we choose an arbitrary order of the vertices at level $T'$, say $(v_1,T'),\ldots, (v_m,T')$ where $m = |Q_{T'}|$. For each vertex at level $T'$ from $v_1$ to $v_m$\,, we start revealing the three sampled neighbours of them at level $T'-1$ one by one. We say that a collision happens at $(v,T')$ if $(w,T'-1)$ was revealed by $(v,T')$ and it was already revealed by another vertex before $v$ in the order at level $T'$ or by $v$ itself. (See Figure~\ref{fig:eg_sprinkling} as an example.) 
\begin{figure*}[h]
    \centering
    \includegraphics[width=0.8\textwidth]{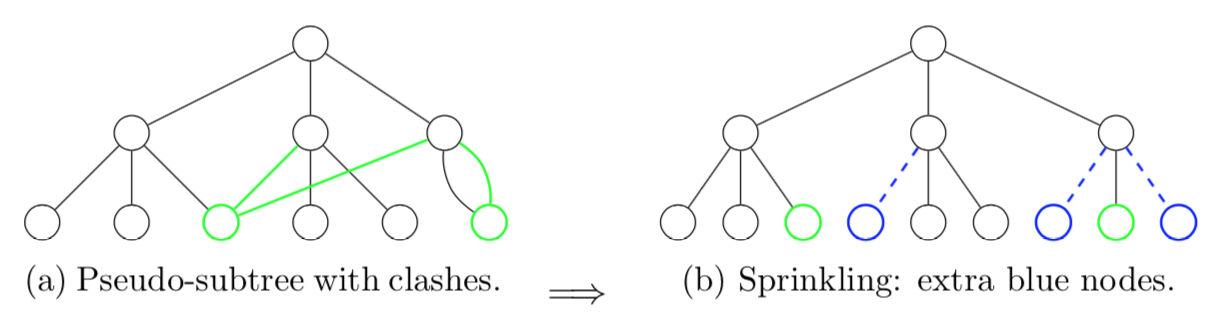}
    \caption{$H$ with 2 levels. Start the Sprinkling process in level 1. Vertices are ordered from left to right.}
	\label{fig:eg_sprinkling}
\end{figure*}
In such a case, we first erase $((v,T'),(w,T'-1))$ from the edge set of $H$\,. Then we add a new vertex to $V(H)$ at level $T'-1$, say $(q,T'-1)$, and a new edge from $(v,T')$ to $(q,T'-1)$\,. Next we set the outdegree of $(q,T'-1)$ to be $0$\,, and set the opinion of $(q,T'-1)$ to be deterministically $1$ (or \textbf{B} in colour language) irrespective of the actual colour of $(w,T'-1)$\,. By applying the Sprinkling model when revealing the neighbours of vertices at level $T'$, we will have a collision-free level, where any two vertices at this level do not have common neighbours. After this we repeat the Sprinkling process on levels $T'-1, T'-2,\ldots$ up to level $1$.  At the end we will have a new voting-DAG $H'$ with $V(H) \subseteq V(H')$, and all vertices in $V(H')\setminus V(H)$ have colour 1 (\textbf{B}) deterministically. Apart from the Sprinkling process, the rest of the colouring process of $H'$ is the same as what we do in $H$. We colour all normal vertices (which are not artificially added) at level $0$ with colour \textbf{B} with probability $1/2-\delta$, otherwise red. As those vertices without a collision also exist in $V(H)$\,, a coupling that uses the same (random) initial colours in both $H$ and $H'$ gives us $X_H(v,t) \leq X_{H'}(v,t)$ for all $(v,t) \in V(H)$, as a result of the extra blue vertices we added in the Sprinkling process. Denote by $\mathcal H'$ the result of the above process applied to the random voting-DAG $\mathcal H$.

Since the Sprinkling process gives a collision-free subgraph from level $0$ to $T'$ in $\mathcal H'$\,, it holds that for $t \in \{0,\ldots, T'\}$, $\{X_{\mathcal H'}(v,t), (v,t) \in \mathcal V(\mathcal H)\}$ are independent random variables. In spite of the fact that their distribution is not identical (and is difficult to compute because of its dependency on several factors, such as the colours at level $0$ and the (random) structure of $\mathcal H$ and thus $\mathcal H'$)\,, we will prove that
$$\Prob(X_{\mathcal H'}(v,t) = \textbf{B}| (v,t) \in V(\mathcal H)) \leq p_t,$$ where $p_t$ satisfies the recursion $p_0 = 1/2-\delta$\,, and 
\begin{align}
p_t 
&\leq (3{p_{t-1}}^2 - 2{p_{t-1}}^3) \cdot ( 1- \varepsilon_{t-1} )^3 \nonumber\\
&\quad \;+ \left( 2 p_{t-1}-{p_{t-1}}^2 \right) \cdot \begin{pmatrix} 3 \\ 1 \end{pmatrix} \varepsilon_{t-1}( 1- \varepsilon_{t-1})^2 \nonumber\\
&\quad \;+ 1\cdot \begin{pmatrix} 3 \\ 2 \end{pmatrix} {\varepsilon_{t-1}}^2 ( 1- \varepsilon_{t-1}) + 1\cdot \begin{pmatrix} 3 \\ 3 \end{pmatrix} {\varepsilon_{t-1}}^3 \,\nonumber\\
&\leq  (3{p_{t-1}}^2 - 2{p_{t-1}}^3)+ 6p_{t-1}\varepsilon_{t-1} + 3\varepsilon_{t-1}^2+\varepsilon_{t-1}^3 \;\,,
\label{eqn:boundProbSP}
\end{align}
where $\varepsilon_{t-1} = 3^{T-t+1}/d$.

The proof goes by induction. `Clearly' the bound applies for any vertex at level $0$\,. Assume it works up to level $t-1$ and consider a vertex $(v,t) \in V(\mathcal H)$ at level $t$\,. The event that $(v,t)$ is coloured by $\textbf{B}$ in $\mathcal H'$ is the same as: vertex $(v,t) \in V(\mathcal H)$ has at least two neighbours that have opinion \textbf{B} at level $t-1$\,, or there is exactly one collision in $(v,t)$ and at least one of its two normal neighbours are \textbf{B}, or there are 2 or 3 collisions in $(v,t)$\,.
Note that at level $t-1$ there are at most $3^{T-t+1}$ vertices, therefore when revealing one neighbour of $(v,t)$, the probability of a collision is at most $3^{T-t+1}/d(v)\leq 3^{T-t+1}/d= \varepsilon_{t-1}$. Then, the expression in~equation~\eqref{eqn:boundProbSP} is obtained by revealing the neighbours of the vertices at level $t$ independently of the order. The first term is the probability that no collision occurs and there are at least two blue vertices out of three normal vertices, the second term represents one collision with at least one blue vertex out of two normal vertices, and the last two terms mean two and three collisions respectively. We summarise the above argument in the following proposition.

\begin{proposition}\label{prop:T'}
Let $G = (V,E)$ be a graph of $n$ vertices. Let $v \in V$ be any vertex and consider $\mathcal H$ the random voting-DAG associated to $v$ of $T$ levels. Let $T'\leq T$, then the opinions at level $T'$ can be majorised by a set of independent opinions where the probability of being \textbf{B} is given by $p_{T'}$ as in equation~\eqref{eqn:boundProbSP}\,, where $\varepsilon_{t-1}=3^{T-t+1}/d$\,.
\end{proposition}

\begin{lemma}
Let $G = (V,E)$ be a graph with minimum degree $d \geq n^{\beta/\log \log n}$ for some $\beta > 0$, and assume the initial opinions are independently $\textbf{B}$ with probability $1/2-\delta$ with $\delta \geq (\log d)^{-C}$ for some $C>0$. Let $v \in V$ be an arbitrary vertex. Then for any $a>0$ there exists $T \geq \lfloor a\log \log d\rfloor$ such that if we consider the random voting-DAG $\mathcal H$ of $T$ levels, then the opinions at level $T-\lfloor a\log \log d\rfloor$ are majorised by a vector of independent opinions where opinion \textbf{B} has probability $o(d^{-1})$.
\end{lemma}
\begin{proof}

The proof consists of considering a voting-DAG of heigh $T = \lfloor a\log \log d\rfloor+1+T_2+T_3$, where $T_2$ and $T_3$ are chosen later. Remember that at level $0$ all vertices have independent opinions with $p_B = 1/2-\delta$. Our proof consists of three steps: i) opinions at level $T_3$ can be majorised by i.i.d. opinions with  $p_B = 1/2-{1}/{2\sqrt 3}$, ii) opinions at level $T_2+T_3$ can be majorised by i.i.d. opinions with $p_B = \polylog (d)/d$, and iii) opinions at level $1+T_2+T_3$ can be majorised  by i.i.d. opinions with $p_B = o(d^{-1})$. 

We first check iii) assuming i) and ii). For that, we ignore all previous levels and consider a voting-DAG of height $h_1 = \lfloor a\log \log d\rfloor+1$ and the colour of the leaves are independently $\mathbf{B}$ with probability $p_0 = \polylog(d)/d$.  From equation~\eqref{eqn:boundProbSP} we have
$$p_1 \leq 3p_0^2+6p_0\varepsilon_0+3\varepsilon_0^2
+\varepsilon_0^3$$
and $\varepsilon_0 = 3^{h_1}/d = \polylog (d)/d$, hence $p_1 = \left({\polylog(d)}/{d}\right)^2 = o(1/d)$. 

Next, we check ii) assuming i). We consider a voting-DAG of height $h_2 = h_1+T_2$, and assume that leaves are independently $\mathbf{B}$ with probability $p_0=1/2-1/(2\sqrt{3})$.  We choose $T_2$ to be \[T_2 = \min\{\min\{t \geq 0: p_t \leq 12\varepsilon_t\},2\log_2 \log d\}\,.\] For $t \leq h_2$ we have that $\varepsilon_t \leq 3^{h_2}/d = 3^{O(\log \log d)}/d = \polylog(d)/d$. Then for $t \in \{1,\ldots, T_2\}$, we have that
\begin{align}
p_t \leq 3p_{t-1}^2+6p_{t-1}\varepsilon_{t-1}+4\varepsilon_{t-1}^2 \leq 4p_{t-1}^2.\label{eqn:random1823hsd}
\end{align}
The first inequality of \eqref{eqn:random1823hsd} is due to equation~\eqref{eqn:boundProbSP}, while the second follows from the fact that for $t \leq T_2$ we have $p_{t-1} > 12\varepsilon_{t-1}$. Iterating the recursion, we have that
$$p_t \leq (4p_0)^{2^t} \leq \left(4\left(\frac{1}2-\frac{1}{2\sqrt{3}} \right)\right)^{2^t} \leq (0.85)^{2^t}.$$

Let $L= \log_2(-\log (d)/\log (0.85))$, and note that $L \leq 2\log_2 \log d$ for large enough $d$. If $L< T_2$ then we would have that $p_L \leq 1/d$ but $\varepsilon_L \geq 3^{h_2-T_2}/d = 3^{h_1}/d = \polylog(d)/d$, which contradicts the fact that $p_L >12\varepsilon_L$. Therefore, we conclude that $T_2 = \min\{t \geq 0: p_t \leq 12\varepsilon_t\} =O(\log \log d)$ and that opinions at level $T_2$ can be majorised by independent opinions with the probability of being blue equal to $12\varepsilon_{T_2} = 12 \times 3^{h_1}/d =  \polylog(d)/d$.

Finally, we check i). We consider a voting-DAG of height $h_3 = h_2+T_3$, and assume the initial opinions are i.i.d  with probability of being blue $p_0 = 1/2-\delta$. We choose $T_3 = \min\{\min\{t \geq 0: \delta_t \geq 1/(2\sqrt{3})\}, C\log \delta^{-1}\}$, where $C$ is a suitable constant greater than $10/\log(5/4)$. Let $\delta_t = 1/2-p_t$. Then, replacing $p_t = 1/2-\delta_t$ in equation~\eqref{eqn:boundProbSP} and noting that $3\varepsilon_{t-1}^2 + \varepsilon_{t-1}^3 \le \varepsilon_{t-1}$ give us
\begin{align}
\delta_t \geq \delta_{t-1}+\left(\frac{1}2 \delta_{t-1}-2\delta_{t-1}^3-4\varepsilon_{t-1}\right) \;\,.
\label{eqn:deltatweqwr}
\end{align}
The function $f(x) = x/2-2x^3$ is such that $f(0) = 0$, and it is increasing from $0$ to $1/(2\sqrt{3})$ where it reaches a local maximum. Note that if $\delta_{t-1} \ge 12 \varepsilon_{t-1}$ and $\delta_{t-1} <1/(2\sqrt{3})$, then equation~\eqref{eqn:deltatweqwr} yields
\begin{align}
&\frac 12 \delta_{t-1} -2\delta_{t-1}^3-4 \varepsilon_{t-1} \nonumber\\
\geq &\delta_{t-1}\left(\frac 12-2\delta_{t-1}^2- \frac{\varepsilon_{t-1}}{\delta_{t-1}} \right) \nonumber\\
\geq &\delta_{t-1}\left(\frac 12-\frac16- \frac{1}{12} \right)  = \frac{\delta_{t-1}}{4}\;\,,\label{eqn:random1231jsf4}
\end{align}
implying that $\delta_t \geq 5\delta_{t-1}/4$. Note that as long as we can apply the previous recursion we have an increasing sequence of $\delta_t$ while $\varepsilon_t$ is decreasing, therefore if $\delta_0\gg \varepsilon_0$ then $\delta_t \gg \varepsilon_t$ for all $t \leq T_3$. To check that $\delta_0 \gg \varepsilon_0$, recall that $\delta_0 =\delta \geq (\log d)^{-C}$ for some constant $C>0$, and  $\varepsilon_0 \leq  3^{h_3}/d = 3^{O(\log \log d+\log \delta^{-1})}/d= (\polylog(d))/d$, then $\varepsilon_0 \ll \delta_0 $. We conclude that $\delta_t>5\delta_{t-1}/4$ for all $t\leq T_3$.
Let $L = \frac{\log(\delta^{-1}/(2\sqrt{3}))}{\log (5/4)}$. Note that 
$\delta_{L} \geq (5/4)^{L}\delta_0 \geq 1/(2\sqrt{3})$, implying that $L\geq \min\{t\geq 0: \delta_t \geq 1/(2\sqrt{3})\}$, therefore by our choice of $C$ in $T_3$ we conclude that $T_3 = \min\{t\geq0: \delta_t \geq 1/(2\sqrt{3})\} = O(\log(\delta^{-1}))$.
\end{proof}

\section{Upper Levels}\label{sec:upperLevels}

From the results in the previous section, we know that the opinions of vertices at level $T'$ (see Proposition~\ref{prop:T'}) in $H$ are majorised by i.i.d. Bernoulli random variables with probability of being $1$ (or colour \textbf{B}) equals $o(d^{-1})$. In this section, we will deal with the levels above $T'$\,.

Now, since there is no need to care about lower levels, we assume that $\mathcal H$ is a voting-DAG of $h+1 = T-T'$ levels with root $(v_0,h)$ and that the vertices at levels $0$ are independently \textbf{B} with probability $o(d^{-1})$, otherwise \textbf{R}. Our strategy to deal with this case is to show that for most realisations of the random voting-DAG $\mathcal H$, the number of vertices at bottom with colour \textbf{B} is too small for the root of $\mathcal H$ to have colour \textbf{B}. We start by supposing that $\mathcal H$ is (deterministically) a ternary tree. 

\begin{lemma}\label{lemma:leastB}
Suppose that $H$ is a ternary tree of $h+1$ levels. Then, if the number of leaves with opinion \textbf{B} is less than $2^h$, then the root has opinion \textbf{R}.
\end{lemma}

\begin{proof}
The statement is equivalent to that if the root is \textbf{B} then there are at least $2^h$ vertices with opinion \textbf{B} at level 0. The result holds easily by noting at least two neighbours of the root have opinion \textbf{B}, and that they are also the root of a (sub)-ternary tree of $H$.
\end{proof}

For the case that the voting-DAG is not a ternary tree, the next lemma establishes that we can find a colouring on a ternary tree that gives the same colour to the root, and the number of \textbf{B} leaves in the ternary tree depends on the number of levels that involve collisions in the DAG.

\begin{lemma}\label{lemma:3-tree}
Let $H$ be a fixed voting-DAG of $h+1$ levels with root $v_0 \in V(G)$. Given a colouring $\xi$ of the vertices at level $0$, there exists a colouring $\xi'$ of the leaves of a ternary tree $H'$ of $h+1$ levels such that the colouring process in $H$ and in $H'$ give the same colour to the root. Moreover, the number of \textbf{B} leaves in $\xi'$ is at most $B_0 \cdot 2^C$ where $B_0$ is the number of \textbf{B} leaves in $\xi$ and $C$ is the number of levels of $H$ that involve at least one collision.
\end{lemma}

\smallskip
\begin{proof}
The proof follows by induction on the number of levels. If the number of levels is 1, then $H$ is a single vertex and the result holds trivially. Suppose the result holds for $h$ levels, we will prove it for $h+1$ levels.  Let $X_H$ be the colouring of $H$ given the colouring $\xi$ of the leaves. Consider the root at level $h$ and let $e_1$, $e_2$ and $e_3$ be its three outgoing edges. We consider two cases: i) at least two of these edges share the same endpoint at level $h-1$ (i.e. a collision at level $h$), or ii) the edges do not share endpoints at level $h-1$ (i.e. level $h$ is collision-free). In the first case i) the opinion of $(v_0,h)$ is determined by the colour of the shared endpoint, say $(v,h-1)$. In this case, we consider a voting-DAG $H'$ of $h+1$ levels. At level $h$ we have the root $(v_0,h)$\,. At level $h-1$ we put two disjoint copies of $H[(v,t)]$ (without sharing vertices), and one ternary tree of $h-1$ levels. Then we connect $(v_0,h)$ with the root of those three sub-graphs. We colour the leaves of $H'$ as follows. In the copies of $H[(v,t)]$ the colours of the vertices are given by the original opinions settling in $\xi$\,, while the leaves of the ternary tree are attached to colour \textbf{R}\,. Note that the colour of the root of $H'$ is the same as the root of $H$ since the colour of the root of $H'$ is determined by the colour of the root of $H[(v,t)]$ (the colour of the root of the ternary tree is irrelevant). By the induction hypothesis, $H[(v,t)]$ can be transformed into a tree with at most $B_0'2^{C'}$ leaves with opinion \textbf{B}, where $B_0'$ and $C'$ are the number of blue leaves and the number of levels involving at least one collision in $H[(v,t)]$, respectively. By the construction of $H'$, all the collisions are represented in the copies of $H[(v,t)]$ which are $C'$. Clearly $C'+1\leq C$ and $B_0'\leq B_0$. Applying the induction hypothesis to the two copies of $H[(v,t)]$, where $H[(v,t)]$ is transformed in a ternary tree such that $B_0'2^{C'}$ leaves have colour \textbf{B}. As $H'$ contains two copies of such graphs, after the induction step we get a ternary tree with $2^{C'}2B_0' = 2^{C'+1}B_0' \leq 2^{C}B_0$ leaves with colour $B$. Case ii) can be done similarly by applying the induction hypothesis to the three vertices in level $1$.
\end{proof}

Finally, we combine the two previous lemmas to show that with w.h.p the root is \textbf{R}. The idea is to show that the number of levels involving collisions in the DAG is not large and therefore a straightforward application of Lemma \ref{lemma:3-tree} and Lemma~\ref{lemma:leastB}, together with the fact that a leaf is \textbf{B} with probability $o(d^{-1})$, tell us the the root of the DAG is \textbf{R} w.h.p.

\begin{lemma}
Consider a random voting-DAG $\mathcal H$ with $h+1$ levels, whose leaves have opinion \textbf{B} with probabiluty $o(d^{-1})$, otherwise $\textbf{R}$. Then, with probability $o(n^{-1})$ the root of $\mathcal H$ is \textbf{B}.
\end{lemma}

\begin{proof}
Let $C_i$ be the indicator random variable taking value $1$ if at least one  clash occurs at level $i$\,. Recall that level $i$ involves a collision if two vertices at level $i$ share a neighbour at level $i-1$\,. Consider the event $E_i = \{\text{there are $m_i$ vertices at level $i$}\}$\,, where $m_i \leq 3^{h-i}$\,, and start revealing the neighbours of the vertices at level $i$ one by one. Then,
\begin{align*}
&\Pr(C_i = 1|E_i) = 1- \Pr (C_i = 0|E_i ) \\
\le &1- \left[ (1-1/d)(1- 2/d) \cdots (1- (m_i-1)/d) \right] \\
\le &1- (1- m_i/d)^{m_i} \le 1- ( 1- m_i \cdot m_i/d) \\
= &{m_i}^2 /d \leq 9^{h}/d\,.
\end{align*}
Denote by $C = \sum_{i=1}^{h} C_i$ the total number of levels that involve at least one collision. Then, as $C_i$ only depends on the out-edges of vertices at level $i$ of $H$\,, we have that $C$ can be majorised by a $Bin(h, 9^h /d)$ random variable.

We first construct a ternary tree $\mathcal H'$ by applying Lemma~\ref{lemma:3-tree} to $\mathcal H$. Let $B$ and $B'$ be the number of leaves with opinion \textbf{B} in $\mathcal H$ and $\mathcal H'$\,, respectively. Then $B' \leq 2^C B$, and by Lemma~\ref{lemma:leastB} it holds that
\begin{align*}
\Prob(\text{root of $\mathcal H$ is blue}) &\leq \Prob(B' \geq 2^h)\\
&\leq \Prob(2^C \cdot B \geq 2^h) \\
&= \Prob(B \geq 2^{h-C}) \;\,,
\end{align*}
and
\begin{align}
\Prob(B \geq 2^{h-C}) &\leq \Prob(B \geq 2^{h-C} |C \geq \frac{h}{2})\Prob(C\geq \frac{h}{2}) \nonumber\\
&\quad\;+ \Prob(B \geq 2^{h-C} |C < \frac{h}{2})\Prob(C<\frac{h}{2})\nonumber\\
&\leq \Prob(C \geq \frac{h}{2})+\Prob(B \geq 2^{h/2})\label{eqn:random20838113} \;\,.
\end{align}
For the first probability of inequality~\eqref{eqn:random20838113}\,, we get
\begin{align}
\Pr(C > \frac{h}{2}) &\le \sum_{k= \lfloor{h/2}\rfloor}^{h} \binom{h}{k} \left(\frac{{9^h}}{d} \right)^k \nonumber \\
&\le \sum_{k= \lfloor{h/2}\rfloor}^{h} \left(\frac{he}{k}\right)^k \left(\frac{{9^h}}{d} \right)^k \nonumber\\
&\le \sum_{k= \lfloor{h/2}\rfloor}^{h} \left(\frac{2e\,9^h}{d} \right)^k \nonumber \\
&\le \left( \frac{2e \,{9^h}}{d} \right)^{\lceil{h/2}\rceil} \sum_{k=0}^{\infty} \left( \frac{2e \,9^h}{d} \right)^k\nonumber\\
&\leq \left(\frac{2e\,9^h}{d}\right)^{h/2} \label{eqn:random39427df} \;\,.
\end{align}
The last step holds as we claim that $2e\,9^h/d\leq 1/2$. To see this, let $h = a\log \log_2(d)$ for some constant $a>0$, then for any $b<1$\,, 
\begin{align*}
2e\,9^h/d &= 2e\exp\{a\log 9 \log (\log_2 d)\}/d \\
&= 2e (\log_2 d)^{a\log 9}/d \leq d^{-b} \;\,.
\end{align*}
This proves the claim. From the previous equation, using $d = n^{\alpha}$ we obtain
\begin{eqnarray}
\Prob(C>\frac{h}{2})\leq d^{-bh/2} = n^{-bh\alpha/2} \;\,.
\end{eqnarray}
We finish by checking that $bh\alpha/2 > 1$\,, in which $h = a\log \alpha\log_2(n)$\,. Then we get
\begin{align}
\frac{bh\alpha}{2} &= \frac{\alpha b a}{2}\log (\alpha\log_2(n)) \nonumber\\
&= \frac{ab}{2}\left(\alpha\log \alpha + \alpha \log \log_2 n \right) \;\,.
\end{align}
For $\alpha > c/\log \log_2 n$\,, we can choose $a$ large enough such that the above quantity is greater than 1, so that $\Prob(C>h/2) = o(n^{-1})$\,. 

\smallskip
For the other term of the inequality~\eqref{eqn:random20838113}\,,
\begin{align*}
\Pr \left(B \ge 2^{h/2}\right) &\le \sum_{k= \lfloor{h/2}\rfloor}^{3^h} \binom{3^h}{k} \left(\frac{3^h}{d} \right)^k \\
&\le \sum_{k= \lfloor{h/2}\rfloor}^{3^h} \left( \frac{3^h\, e}{\lceil{h/2}\rceil } \right)^k \left(\frac{3^h}{d} \right)^k \\
&\le \left( \frac{2e \, {9^h}}{d \, h} \right)^{\lfloor{h/2}\rfloor} \; \sum_{k=0}^{\infty} \left( \frac{2e \, 9^h}{d\,h} \right)^k\nonumber \\
&\leq \left(\frac{2e9^h}{d}\right)^{h/2} \;\,.
\end{align*}
	
\smallskip
The last step holds as long as $2e 9^h/(d h)= o(1)$\,, which was showed before. Note that we already demonstrated that $\left(2e \,9^h/d\right)^{h/2}  = o(n^{-1})$, then we conclude that $\Prob(B > 2^{h/2}) = o(n^{-1})$.
\end{proof}

\bibliography{MyBibFile}
\bibliographystyle{plain}
\end{document}